\documentclass[12pt]{article}
\usepackage[margin=1in]{geometry}
\usepackage{amsthm} 

\usepackage[utf8]{inputenc}
\usepackage[algo2e]{algorithm2e} 
\usepackage{textgreek}
\usepackage{textcomp}
\usepackage{amssymb}
\usepackage{subfigure}
\usepackage{amsmath}

\usepackage{color}
\usepackage{footnote}
\usepackage{algorithm}
\usepackage{algpseudocode}
\usepackage{algorithmicx}
\usepackage{multirow} 
\usepackage{booktabs}
\usepackage{graphicx}
\usepackage{amssymb}
\usepackage{amsbsy}
\usepackage{array}
\usepackage{longtable}
\usepackage{epstopdf}
\usepackage{pbox}
\usepackage{breqn}
\usepackage{mathrsfs}
\usepackage{multicol}
\usepackage{supertabular}
\usepackage{enumerate}
\usepackage{url}
\usepackage[justification=centering]{caption}
\usepackage{comment}
\newtheorem{thm}{Theorem}

\newtheorem{prp}{Proposition}

\title{\LARGE \bf
Latency-Robust Control of High-Speed Signal-Free Intersections
}

\author{Yang Liu, Zev Nicolai-Scanio, Zhong-Ping Jiang and Li Jin
\thanks{This work was in part supported by NYU Tandon School of Engineering, C2SMART University Transportation Center, and US NSF grant EPCN-1903781.}
\thanks{The authors are with the Tandon School of Engineering, New York University, Brooklyn, New York, USA.
Y. Liu is also with the Department of Automation, Shanghai Jiao Tong University, Shanghai, China.
Z. Nicolai-Scanio is also with the John A. Paulson School of Engineering and Applied Sciences, Harvard University, Cambridge, Massachusetts, USA.
Emails:  liuyang@nyu.edu, zn2054@nyu.edu,
zjiang@nyu.edu,
lijin@nyu.edu}
}

\begin{document}
\maketitle

\begin{abstract}
High-speed signal-free intersections are a novel urban traffic operations enabled by connected and autonomous vehicles.
However, the impact of communication latency on intersection performance has not been well understood.
In this paper, we consider vehicle coordination at signal-free intersections with latency. We focus on two questions: (i) how to ensure latency-resiliency of the coordination algorithm, and (ii) how latency affects the intersection's capacity.
We consider a trajectory-based model with bounded speed uncertainties. Latency leads to uncertain state observation.
We propose a piecewise-linear control law that ensures safety (avoidance of interference) as long as the initial condition is safe.
We also analytically quantify the throughput that the proposed control can attain in the face of latency.
\end{abstract}


\section{Introduction}

Recent progress in the technology of connected and autonomous vehicles (CAVs) motivates high-speed signal-free intersections \cite{dresner2008multiagent,wu2012cooperative}.
Conventional intersections are either signalized or regulated by stop signs.
With the help of vehicle-to-vehicle/infrastructure (V2V/V2I) connectivity, CAVs can cross intersections at high speeds; such operations are infeasible with human drivers due to high response time and lack of V2V communications \cite{calvert2018traffic}.
The safety and efficiency of signal-free intersections heavily rely on the quality of V2V/V2I communications, which is not always perfect due to power, computation, and bandwidth limitations and hardware disruptions \cite{kamini2010vanet}.
It is widely known that even small delays can cause performance deterioration or sometimes instability in dynamical systems \cite{karafyllis2011stability}.
However, very limited work has been done on analysis of the impact of communication latency on performance of such intersections.

In this paper, we consider the vehicle coordination problem for a signal-free intersection subject to a non-zero latency; see Fig.~\ref{fig:intersection}.
\begin{figure}[hbt]
    \centering
    \includegraphics[width=0.3\textwidth]{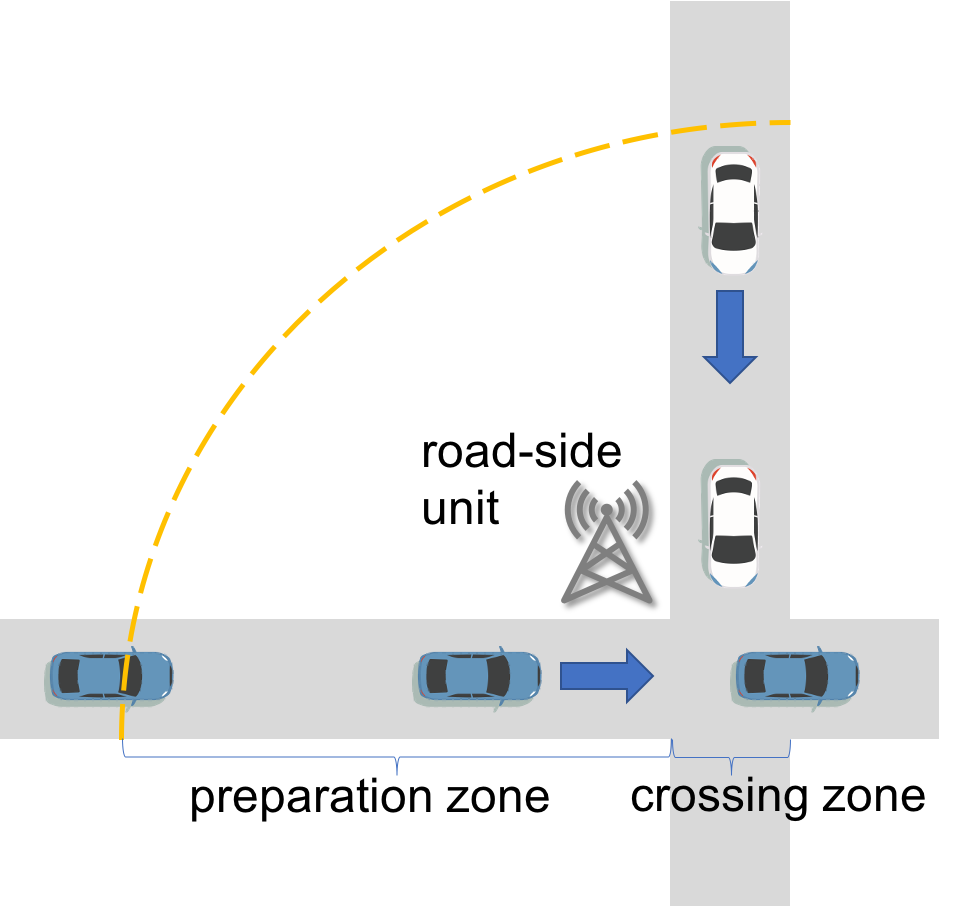}
    \caption{A road-side unit coordinates vehicles in the preparation zone before they entering the crossing zone.}
    \label{fig:intersection}
\end{figure}
The main questions that we ask are:
\begin{enumerate}
    \item How to design a vehicle coordination algorithm that is robust against communication latency?
    \item How to quantify the relation between communication latency and intersection capacity?
\end{enumerate}
To address the above questions, we consider a trajectory-based model for vehicles in a neighborhood of an intersection with a centralized controller.
Latency leads to delayed state observation, and vehicle trajectories are subject to bounded uncertainty.
We design a robust vehicle coordination algorithm that ensures all-time safety as long as the initial condition is within a safety set (Theorem~\ref{thm:1}).
The proposed algorithm also minimizes vehicle traverse time under the safety constraint (Theorem~\ref{thm:2}).
We also derive an analytical relation between key parameters, including latency, coordination step size, and speed uncertainty, and intersection capacity (Proposition~\ref{prp:capacity}).

Recently, novel models and methods for signal-free intersections with ideal V2V/V2I connectivity have been proposed \cite{zhu2015linear,malikopoulos2018decentralized,xu2018distributed,miculescu2019polling,mirheli2019consensus,xu2019cooperative};
the proposed control algorithms have been well validated in the absence of latency or other types of communication imperfections.
However, very limited results are available for the latency-prone setting.
Although latency has been studied by the networking community \cite{sun2015radio,dey2016vehicle} and the control community \cite{karafyllis2011stability}, we still lack a connection between communication performance and transportation performance. This gap motivates the application of robust control ideas to vehicle coordination.
There exists robust optimization/control-based methods for low-level trajectory planning \cite{kuwata2010cooperative,ploeg2014graceful,li2018fault}, but these approaches are not directly applicable to intersection control, which is at a higher level.
In addition, very limited results are available for quantifying the impact of communication latency on key transportation performance metrics such as capacity (throughput).

To address these challenges, we consider a trajectory-based model.
A road-side unit (centralized controller) collects real-time kinematic information and controls the speed of all vehicles in a point-to-point manner.
To account for the uncertainty of vehicle trajectories, we consider a set-valued update mapping, which means that the realized speed of a vehicle will fall in a neighborhood of the specified speed rather than exactly at it.
Due to such uncertainty, latency will compromise the accuracy of state observation.
We propose a robust algorithm that controls the speed of vehicles and ensures safety in the face of latency.
The algorithm can be written as an explicit control law and is thus computationally efficient and easy to implement.

The main results of this paper characterize properties and performance of the proposed robust vehicle coordination algorithm.
First, we show that for any initial condition in a safety set, the proposed algorithm ensures safety for all positive times (Theorem~\ref{thm:1} and~\ref{thm:2}).
Second, we quantify the throughput that can be attained by the proposed algorithm (Proposition~\ref{prp:capacity}).

The rest of this paper is organized as follows.
Section~\ref{sec_model} provides the modeling of the intersection and the formulation of the problem.
In Section~\ref{sec_control} we provides the construction of the estimator and the controller.
Then in Section~\ref{sec_prop} we study the properties of the proposed controller.
Finally we conclude the paper in Section~\ref{sec_conclusion}.
\section{Modeling}\label{sec_model}%

Consider the neighborhood of the two-direction intersection in Fig.~\ref{fig:intersection}.
We label the two orthogonal routes as 1 and 2. 
Suppose that there are $n_1$ (resp. $n_2$) vehicles on route 1 (resp. 2).
The state of the intersection is described by $x=(x^1,x^2)\in\mathcal X:=[-L,L]^{n_1+n_2}$ and $v=(v^1,v^2)\in\mathcal U:=[0,\bar v]^{n_1+n_2}$, where
\begin{enumerate}
    \item $x^k=[x^k_1(t)\ x^k_2(t) \cdots\ x^k_{n_k}(t)]^T$ denotes the positions of vehicles on route $k$ at time $t\in\mathbb Z_{\ge0}$,
    
    \item $v^k=[v^k_1(t)\ v^k_2(t) \cdots\ v^k_{n_k}(t)]^T$ denotes the speeds of vehicles on route $k$ at time $t\in\mathbb Z_{\ge0}$,
    
    \item $L$ is the radius of the intersection region,
    
    \item $\bar v$ is the maximal allowable speed.
\end{enumerate}
We index the vehicles by the order of entering the neighborhood; i.e. vehicle 0 on route $k$ is the last vehicle that arrived in the neighborhood before vehicle 1 on route $k$.
When we consider vehicle $i$, we assume that the trajectories of all vehicles in front of vehicle $i$ have been optimized and fixed. 
Our task here is to optimize the trajectory $x_i^j(t)$ and $v_i^j(t)$ of vehicle 0 under the constraint due to other vehicles' trajectories recursively.

The decision variables is the sequence of \emph{target speeds} $\{u(t);t\in\mathbb Z_{\ge0}\}$.
CAVs attempt to track the target speed.
Specifically, $v(t)$ is jointly determined by $v(t-1)$ and $u(t)$ as follows.
For $x\in\mathbb R_{\ge0}$ and $r\in\mathbb R_{\ge0}$, let $B_r(x)$ be the non-negative neighborhood
$$
B_r(x):=[x-r,x+r]\cap\mathbb R_{\ge0}.
$$
Then, for each $k$ and each $i$,
\begin{align*}
    &v^k_i(t)\in 
    Q^k_i(v(t-1),u(t)) \\
    &:=\begin{cases}
    B_\epsilon(v(t-1)-\bar a) 
    & u(t)<v(t-1)-\bar a,
    \\
    B_{\epsilon}(u(t))
    & u(t)\in B_{\bar a}(v(t-1)),
    \\
    B_\epsilon(v(t-1)+\bar a)
    & u(t)>v(t-1)-\bar a.
    \end{cases}
\end{align*}
The neighborhood size $\epsilon$ depends on the type of vehicles: autonomous vehicles (AVs) are associated with a smaller size than non-autonomous vehicles (non-AVs).

We assume that the control input $v(t)$ is determined by a control law $\mu:\mathcal X\times\mathcal U\to[0,\bar v]^{2n}$ such that, in the ideal setting without latency,
$$
u_i(t)=\mu_i(x_i(t-1),v_i(t-1),x_{i-1}(t-1),v_{i-1}(t-1)).
$$
In the presence of a latency $\theta>0$, state estimation becomes set-valued.
\begin{figure}[hbt]
    \centering
    \includegraphics[trim=2cm 7cm 0 3cm, clip,width=0.6\textwidth]{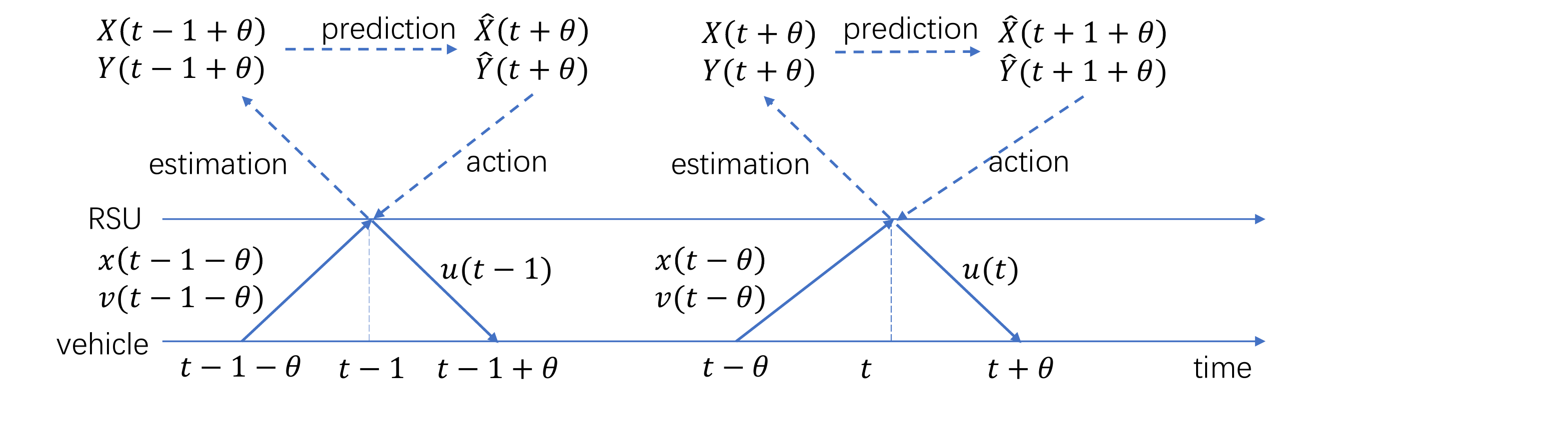}
    \caption{Latency delays sending and actuation.}
    \label{fig:latency}
\end{figure}
That is, the reported state $(\hat x(t-1),\hat v(t-1))$ that the RSU receives at time $t-1$ indicates that
\begin{align*}
    &x(t-1-\theta)=\hat x(t-1),\\
    &v(t-1-\theta)=\hat v(t-1).
\end{align*}
The RSU estimates the actual state with sets
\begin{align*}
    &x(t)\in X_\theta(\hat x(t),\hat v(t)),\\
    &v(t)\in V_\theta(\hat x(t),\hat v(t)),
\end{align*}
where $X_\theta$ and $V_\theta$ are estimators to be designed in Section~\ref{sec_estimate}.
When vehicles receive the instruction $u(t)$ at time $t+\theta$, the speeds are updated by
\begin{align*}
    v(t+\theta)\in Q\Big(V_{2\theta}(\hat x(t),\hat v(t)),u(t)\Big),
\end{align*}
where $Q$ is the set-valued version of the speed update function.
The positions are updated by
\begin{align*}
    x(t+\theta+1)=x(t)+\frac{1}{2}\Big(v(t+\theta)+v(t+\theta+1)\Big).
\end{align*}
Fig.~\ref{fig:latency} illustrates the information flow with latency.

The control law must satisfy the \emph{safety constraint}, which lower-bounds the headway between vehicles:

\noindent(i) $i \neq j,k=1,2$:
\begin{align*}
    &
    |x_i^k(t)-x_j^k(t)|
    \begin{cases}
    \ge hv_j^k(t)
    & x_i^k(t)>x_j^k(t),\\
    \ge hv_i^k(t)
    & x_i^k(t)<x_j^k(t),
    \end{cases}
\end{align*}
(ii) $k_1\neq k_2, x_i^{k_1}\le -R,x_j^{k_2}\le -R$:
\begin{align*}
    &
    |x_i^{k_1}(t)-x_j^{k_2}(t)| 
    \begin{cases}
    \ge hv_j^{k_1}(t)
    & x_i^{k_1}(t)>x_j^{k_2}(t),\\
    \ge hv_i^{k_1}(t)
    & x_i^{k_1}(t)<x_j^{k_2}(t),
    \end{cases}
\end{align*}
where $R$ is the radius of the interference region; see Fig.~\ref{fig:intersection}.
\section{Control design}\label{sec_control}

The objective of this section is to design a controller that coordinates the speeds of vehicles in the preparation zone. As an intermediate step, we will develop an estimator that is robust to latency.

\subsection{Estimator design} \label{sec_estimate}

After the RSU receives the state variables of the vehicle at time $t$, then we need to estimate the position and speed information at time $t+\theta$.
The design problem can be formulated as follows: 

{\bf Given}: model parameters $\theta, \epsilon$ and observed variables $x(t-\theta), v(t-\theta), u(t-1), u(t)$.\\
\indent {\bf Estimate}: $x(t+\theta)$ and $v(t+\theta)$.

Under the assumption that the latency $\theta$ is much less than the interval of a time step, then if the vehicle follow the control input correctly, we have
\begin{align*}
    &u(t-1)-\epsilon \le v(t+\theta) \le u(t-1)+\epsilon,
    \\
    &x(t-\theta)+\theta(v(t-\theta)+u(t-1)-\epsilon)
    \le
    x(t+\theta)
    \le \\
    &x(t-\theta)+\theta(v(t-\theta)+u(t-1)+\epsilon).
\end{align*}
Therefore, we construct estimators as follows:
\begin{align*}
     &X_{2\theta}(x,v)\\
     &=\left\{x'\in\mathcal X:
     \begin{array}{l l}
     &x+\theta(v+u(t-1)-\epsilon)
    \le 
    x'
    \le \\
   & x+\theta(v+u(t-1)+\epsilon)
   \end{array}
    \right\},
    \\
    &V_{2\theta}(x,v)=\{v'\in\mathcal V:
    u(t-1)-\epsilon
    \le
    v'
    \le u(t-1)+\epsilon\}.
\end{align*}
Hence, when the control input $u(t)$ is received by the vehicles, the speeds are predicted by the predictor
\begin{align*}
    \tilde V(t+\theta+1)=Q\Big(V_{2\theta}(\hat x(t),\hat v(t)),u(t)\Big),
\end{align*}
and the positions are updated by
\begin{align*}
     \tilde X(t+\theta+1)=
     &X_{2\theta}\Big(\hat x(t),\hat v(t)\Big)+\frac{1}{2} \delta\Big(V_{2\theta}(\hat x(t),\hat v(t))+ \\ &Q(V_{2\theta}(\hat x(t),\hat v(t)),u(t))\Big);
\end{align*}
where $\delta$ denotes the interval of a time step.

We assume that the control input $u(t)$ can be reached under the limitation of the acceleration, then we have that
\begin{align*}
    \tilde V(t+\theta+1)=\{v'\in\mathcal V:
    u(t)-\epsilon
    \le
    v'
    \le u(t)+\epsilon\}.
\end{align*}
Therefore, we can get the interval of $ \tilde X(t+\theta+1)$ as follow:
\begin{align}\label{eq:interval} \nonumber
    &\Delta=\Big|X_{2\theta}(\hat x(t),\hat v(t))\Big|+\frac{1}{2} \delta\Big(\Big|V_{2\theta}(\hat x(t),\hat v(t))\Big|+\\
    &\Big| \tilde V(t+\theta+1)\Big|\Big)
    =2\theta\epsilon+\frac{1}{2} \delta(2\epsilon+2\epsilon)
    =2\epsilon(\theta+\delta).
\end{align}
\subsection{Controller design}\label{sec_controller}

The objective of the control problem is to coordinate the trajectories of adjacent vehicles to discharge vehicles as fast as possible while avoiding interference.
The control input is the target speed assigned to each vehicle at each time step.
In the preparation zone, vehicles on two routes are coordinated separately. For route $k$, the controller $\mu$ must generate control inputs $u(t)$ at each time $t$ that satisfy the safety constraint.
The control problems can be formulated as follows:

{\bf Given}: model parameters $\delta$ (time step size), $\theta$ (latency), $\epsilon$ (speed uncertainty), $\bar a$ (maximum acceleration) and observed states $\hat x(t)$ (positions) and $\hat v(t)$ (speeds).

{\bf Determine}: $u(t)$ (target speed).

Note that since the uncertainty is small compared with the regulating capacity of the vehicle, we have $\epsilon \leq \bar a\delta$
Then for all $i,j,k_1,k_2$, the safety constraint is
\begin{align*}
    &\Big|\tilde X_i^{k_1}(t+1)-\tilde X_j^{k_2}(t+1)\Big|
    \ge \\
    &\begin{cases}
    h\tilde V_j^{k_2}(t+1)
    &k_1=k_2,\ X_i^{k_1}(t+1)>\tilde X_j^{k_2}(t+1),
    \\
    h\tilde V_i^{k_1}(t+1)
    &k_1=k_2,\ X_i^{k_1}(t+1)<\tilde X_j^{k_2}(t+1),
    \\
    \bar h\tilde V_j^{k_2}(t+1)
    &k_1\ne k_2,\ X_i^{k_1}(t+1)>\tilde X_j^{k_2}(t+1),
    \\
    \bar h\tilde V_i^{k_1}(t+1)
    &k_1\ne k_2,\ X_i^{k_1}(t+1)<\tilde X_j^{k_2}(t+1).
    \end{cases}
\end{align*}

Consider the vehicles on the same road, indexed by $1,\dots,n$.
For the vehicle indexed by $n$, which is the last vehicle entering the road, we set $u_n(t)$ to the maximal allowable speed or the speed the vehicle can reach with the maximal acceleration.
Then for $i=1,\dots,n-1$ we have the objective function:
\begin{align} \nonumber
\lambda_i(t)=&\tilde X_{(i-1)min}(t+\theta+1)-\tilde X_{imax}(t+\theta+1)-\delta_{min}
\\ \nonumber
=&\theta\Big(\hat v_{i-1}(t)-\hat v_i(t)+u_{i-1}(t-1)-u_i(t-1)\Big)-2\delta\epsilon \\ \nonumber
&-2\theta\epsilon+\hat x_{i-1}(t)-\hat x_i(t)+\frac{1}{2}\delta\Big(u_{i-1}(t-1) \\ \label{eq:function}
&-u_i(t-1)+u_{i-1}(t)-u_i(t)\Big)-h(u_i(t)),
\end{align}
where $\delta_{min}=hu_i(t)$ is the minimum allowable distance between two vehicles.

Consider the feasibility of the speed under the limitation of the acceleration, we have
\begin{align*}
    u(t)\in [u(t-1)-a\delta+\epsilon,u(t-1)+a\delta-\epsilon]
\end{align*}
denoted by $U_{feasible}(t)$.
Then we construct the controller $\mu$ with the following law to obtain $u(t)$:
\begin{align*}
u_i(t)&=\mu_i(x_i(t-\theta),v_i(t-\theta),x_{i-1}(t-\theta),v_{i-1}(t-\theta))\\
&=\mathop{\arg \min}\limits_{u(t)\in U_{feasible}(t)}\{\lambda_i(t)|\lambda_i(t) \geq 0\}.
\end{align*}

\section{Properties of the Proposed Controller}\label{sec_prop}

In this section, we study key properties of the controller presented in the previous section.
First, we study the criterion for the proposed controller to exist (i.e. to satisfy the safety constraint) at a certain time $t$.
Second, we determine a set of initial conditions that guarantee safety for all times.
Finally, we quantify the throughput that the proposed controller can attain.

\subsection{Safety criteria for one time step}
In this subsection, we discuss the condition to ensure the safety and if objective function $\lambda_i(t)$ can be set to 0.
\begin{thm}\label{thm:1}
The existence of the safety control input depends on the following conditions:
\begin{enumerate}
    \item 
    There exists a safe control input $u(t)$ if and only if it satisfies the following condition:
    \begin{align*} \nonumber
        &\hat x_i(t)-\hat x_{i-1}(t)+\theta\Big(\hat v_i(t)-\hat v_{i-1}(t)\Big)+(\theta+\delta+h)\\
        &\times u_i(t-1)
        -(\theta+\frac{\delta}{2})u_{i-1}(t-1)-\frac{1}{2}\delta u_{i-1}(t)+2\epsilon\theta  \\ \tag{\text{condition 1}}
        &+\frac{3}{2}\delta\epsilon-\frac{1}{2}\bar a\delta^2-h\bar a\delta+h\epsilon \le 0.
    \end{align*}
    \item
    Furthermore, there exists a safe control input $u(t)$ such that vehicle $i$ is able to keep the minimal allowable distance from the leading vehicle if and only if condition 1 and the following condition hold:
    \begin{align*} \nonumber
        &\hat x_i(t)-\hat x_{i-1}(t)+\theta\Big(\hat v_i(t)-\hat v_{i-1}(t)\Big)+(\theta+\delta+\\
        &h)u_i(t-1)-(\theta+\frac{\delta}{2})u_{i-1}(t-1) -\frac{1}{2}\delta u_{i-1}(t)+2\epsilon\theta+\\ \tag{\text{condition 2}}
        &\frac{1}{2}\delta\epsilon+\frac{1}{2}\bar a\delta^2+h\bar a\delta-h\epsilon \ge 0.
    \end{align*}
    \end{enumerate}
\end{thm}

\begin{proof}
\emph{Condition 1.}

($\Rightarrow$)
Suppose that there exists a control input $u(t)$ such that vehicle $i$ is able to keep the allowable distance from the leading vehicle, then we have
\begin{align*}
    \max{\tilde X_i(t+1+\theta)}\le \min{\tilde X_{i-1}(t+1+\theta)-h\cdot  u_i(t)};
\end{align*}
that is
\begin{align*}
    &\hat x_i(t)+\theta(\hat v_i(t)+u_i(t)+\epsilon)+\frac{\delta}{2}(u_i(t-1)+u_i(t)+2\epsilon) \\
    &\le  \hat x_{i-1}(t)+\theta(\hat v_{i-1}(t)+u_{i-1}(t-1)-\epsilon)+\frac{1}{2}\delta(u_{i-1}(t-1)\\
    &\quad+u_{i-1}(t)-2\epsilon)-h\cdot u_i(t).
\end{align*}
Considering the limitation of deceleration, we have
$
    u_i(t)\ge u_i(t-1)-\bar a\cdot \delta+\epsilon.
$
Therefore, we have that
\begin{align*}
    &\hat x_i(t)+\theta(\hat v_i(t)+u_i(t-1)+\epsilon)+(u_i(t-1)+\epsilon)\delta \\ &\quad-\frac{1}{2}(\bar a-\frac{\epsilon}{\delta})\delta^2 \\
    &\le\hat x_{i-1}(t)+\theta(\hat v_{i-1}(t)+u_{i-1}(t-1)-\epsilon)\\
    &\quad+\frac{1}{2}\delta(u_{i-1}(t-1)+u_{i-1}(t)-2\epsilon)-h\cdot u_i(t) \\
    &\le\hat x_{i-1}(t)+\theta(\hat v_{i-1}(t)+u_{i-1}(t-1)-\epsilon)+\frac{1}{2}\delta\\
    &\quad\times(u_{i-1}(t-1)+u_{i-1}(t)-2\epsilon)-h(u_i(t-1)-\bar a\delta+\epsilon),
\end{align*}
which implies
\begin{align*}
     &\hat x_i(t)-\hat x_{i-1}(t)+\theta\Big(\hat v_i(t)-\hat v_{i-1}(t)\Big)+(\theta+\delta+h)u_i(t-1)\\
     &-(\theta+\frac{\delta}{2})u_{i-1}(t-1)-\frac{1}{2}\delta u_{i-1}(t)+2\epsilon\theta+\frac{3}{2}\delta\epsilon-\frac{1}{2}\bar a\delta^2\\
     &-h\bar a\delta+h\epsilon \le 0.
\end{align*}
($\Leftarrow$)
When condition 1 holds, if $u_i(t)=u_i(t-1)+\epsilon-\bar a\delta$, we have
\begin{align*}
     &\max \tilde X_i(t+1+\theta)=\hat x_i(t)+\theta(\hat v_i(t)+u_i(t-1)+\epsilon)\\
     &\quad+\frac{\delta}{2}(2u_i(t-1)+3\epsilon-\bar a\delta)\\
     &=\hat x_i(t)+\theta(\hat v_i(t)+u_i(t-1)+\epsilon) +(u_i(t-1)+\epsilon)\delta\\
     &\quad-\frac{1}{2}(\bar a-\frac{\epsilon}{\delta})\cdot \delta^2 \\
     &\le \hat x_{i-1}(t)+\theta(\hat v_{i-1}(t)+u_{i-1}(t-1)-\epsilon)+\frac{1}{2}\delta(u_{i-1}(t-1)\\
     &\quad+u_{i-1}(t)-2\epsilon)-h(u_i(t-1)+\epsilon-\bar a\cdot \delta) \\
     &= \min \tilde X_{i-1}(t+1+\theta)-h\cdot u_i(t).
\end{align*}

\noindent\emph{Condition 2.}

($\Rightarrow$)
Suppose that there exists a control input $u(t)$ such that vehicle $i$ is able to keep the minimal allowable distance from the leading vehicle, then we have
\begin{align*}
    \max{\tilde X_i(t+1+\theta)}= \min{\tilde X_{i-1}(t+1+\theta)-h\cdot  u_i(t)},
\end{align*}
that is
\begin{align*}
    &\hat x_i(t)+\theta(\hat v_i(t)+u_i(t)+\epsilon)+\frac{\delta}{2}(u_i(t-1)+u_i(t)+2\epsilon) \\
    = & \hat x_{i-1}(t)+\theta(\hat v_{i-1}(t)+u_{i-1}(t-1)-\epsilon)+\frac{1}{2}\delta(u_{i-1}(t-1)\\
    &+u_{i-1}(t)-2\epsilon)-h\cdot u_i(t).
\end{align*}
Considering the limitation of acceleration, we have
\begin{align*}
    u_i(t)\le u_i(t-1)+\bar a\cdot \delta-\epsilon.
\end{align*}
Therefore, we have that
\begin{align*}
    &\hat x_i(t)+\theta(\hat v_i(t)+u_i(t-1)+\epsilon) +(u_i(t-1)+\epsilon)\delta+\\
    &\frac{1}{2}(\bar a-\frac{\epsilon}{\delta})\cdot \delta^2 \\
    \ge&\hat x_i(t)+\theta(\hat v_i(t)+u_i(t-1)+\epsilon)+\frac{\delta}{2}(u_i(t-1)+u_i(t)+2\epsilon) \\
    =&\hat x_{i-1}(t)+\theta(\hat v_{i-1}(t)+u_{i-1}(t-1)-\epsilon)+\frac{1}{2}\delta(u_{i-1}(t-1)\\
    &+u_{i-1}(t)-2\epsilon)-h\cdot u_i(t) \\
    \ge&\hat x_{i-1}(t)+\theta(\hat v_{i-1}(t)+u_{i-1}(t-1)-\epsilon)+\frac{1}{2}\delta(u_{i-1}(t-1)\\
    &+u_{i-1}(t)-2\epsilon)-h(u_i(t-1)-\epsilon+\bar a\cdot \delta),
\end{align*}
that is
\begin{align*}
    &\hat x_i(t)-\hat x_{i-1}(t)+\theta\Big(\hat v_i(t)-\hat v_{i-1}(t)\Big)+(\theta+\delta+h)u_i(t-1)\\
    &-(\theta+\frac{\delta}{2})u_{i-1}(t-1)-\frac{1}{2}\delta u_{i-1}(t)+2\epsilon\theta+\frac{1}{2}\delta\epsilon+\frac{1}{2}\bar a\delta^2\\
    &+h\bar a\delta-h\epsilon \ge 0.
\end{align*}
($\Leftarrow$)
We have that $\lambda_i(t)=\min \tilde X_{i-1}(t+1+\theta)-\max \tilde X_i(t+1+\theta)-h\cdot u_i(t)$.
By Equation~(\ref{eq:function}) we have $\lambda_i(t)$ decreases monotonically with $u_i(t)$.
Then we have
\begin{align*}
    \lambda_i(t)\begin{cases}
    \le 0,\ u_i(t)=u_i(t-1)-\bar a\delta+\epsilon, \\
    \ge 0,\ u_i(t)=u_i(t-1)+\bar a\delta-\epsilon.
    \end{cases}
\end{align*}
Therefore, there must exists a control input $u_i(t)$ such that vehicle $i$ is able to just keep the minimal allowable distance from its leading vehicle.
\end{proof}

Intuitively, if condition 1 in Theorem~\ref{thm:1} holds, then we have that when the vehicle takes the maximal deceleration, it can keep a safe distance from its leading vehicle.
Then we have that under the limitation of the deceleration, we can find a feasible control input $u(t)$ such that the vehicle satisfies the safety constraint (i.e. it can keep the safe distance from its leading vehicle).
If condition 2 in Theorem~\ref{thm:1} holds, then we have that when the vehicle takes the maximal acceleration, it can keep a distance from its leading vehicle which is less than the minimal allowable distance.
Then we know that under the limitation of the acceleration, we can find a feasible control input $u(t)$ such that the vehicle can keep a distance from its leading vehicle which is not more than the minimal allowable distance.
Therefore, if both condition 1 and condition 2 hold, then we can certainly find a control input $u(t)$ such that the vehicle can just keep the minimal allowable distance from its leading vehicle, which leads to the best capacity.

Note that when condition 1 holds, we have
\begin{align*}
    &(\theta+\delta+h)u_i(t-1)\\
    &\le \hat x_{i-1}(t)-\hat x_i(t) +\theta\Big(\hat v_{i-1}(t)+u_{i-1}(t-1)\Big)\\
    &\quad-\theta
    \Big(
    \hat v_i(t)+\epsilon\Big)+\frac{1}{2}\delta\Big(u_{i-1}(t-1)+u_{i-1}(t)\\
    &\quad-2\epsilon\Big)-\epsilon\delta+h(\bar a\delta-\epsilon)+\frac{1}{2}(\bar a-\frac{\epsilon}{\delta})\delta^2 \\
    &=A+h(\bar a\delta-\epsilon)+\frac{1}{2}(\bar a-\frac{\epsilon}{\delta})\delta^2,
\end{align*}
where $A=\hat x_{i-1}(t)-\hat x_i(t) +\theta\Big(\hat v_{i-1}(t)+u_{i-1}(t-1)\Big)-\theta
    \Big(
    \hat v_i(t)+\epsilon\Big)+\frac{1}{2}\delta\Big(u_{i-1}(t-1)+u_{i-1}(t)-2\epsilon\Big)-\epsilon\delta$.

When condition 2 holds, we have
\begin{align*}
    &(\theta+\delta+h)u_i(t-1)\\
    \le& \hat x_{i-1}(t)-\hat x_i(t) +\theta\Big(\hat v_{i-1}(t)+u_{i-1}(t-1)\Big)-\\
    &\theta
    \Big(
    \hat v_i(t)+\epsilon\Big)+\frac{1}{2}\delta\Big(u_{i-1}(t-1)+u_{i-1}(t)\\
    &-2\epsilon\Big)-\epsilon\delta-h(\bar a\delta-\epsilon)-\frac{1}{2}(\bar a-\frac{\epsilon}{\delta})\delta^2 \\
    =&A-h(\bar a\delta-\epsilon)-\frac{1}{2}(\bar a-\frac{\epsilon}{\delta})\delta^2.
\end{align*}
Since $\epsilon$ is much smaller than the maximal change in speed every time step, therefore, we have $h(\bar a\delta-\epsilon)+\frac{1}{2}(\bar a-\frac{\epsilon}{\delta})\delta^2>0$.
Then we know that the set of the possible control input satisfying both conditions 1 and 2 is not an empty set.

\subsection{Safety criterion for time series}
Now, we extend the one-step results in the previous subsection to the case of time series. Specifically, we derive a criterion for initial condition to ensure existence of safe control inputs in all subsequent steps.

\begin{thm}\label{thm:2}
Consider the observed initial state $\hat x(1),\hat v(t)$ and initial control input $u(0)$.
\begin{enumerate}
    \item
    There exists a safe input for all subsequent times if the following conditions hold:\\
    (a) For $t=1$ and for any $i=1,\dots,n$,
    \begin{align*}
        &\hat x_i(1)-\hat x_{i-1}(1)+\theta\Big(\hat v_i(1)-\hat v_{i-1}(1)\Big)+(\theta+\delta+\\
        &h)u_i(0)-(\theta+\frac{\delta}{2})u_{i-1}(0)-\frac{1}{2}\delta u_{i-1}(1)\\
        &+2\epsilon\theta+\frac{3}{2}\delta\epsilon-\frac{1}{2}\bar a\delta^2-h\bar a\delta+h\epsilon \le 0.
    \end{align*}
    (b) For $t=2,3,\ldots$ and for any $i=1,\dots,n$,
    \begin{align} \nonumber
        &(\frac{3}{2}\delta+h)u_i(t)-(\frac{\delta}{2}+h)u_i(i-1)-\delta u_{i-1}(t)+\\
        &\frac{\delta}{2}(\bar a\delta+\epsilon)\le 0. \tag{\text{condition 3}}
    \end{align}
    \item
    Furthermore, there exists a safe input such that vehicle $i$ is able to keep the minimal allowable distance from vehicle $i-1$ at all times if\\
    (a) For $t=1$,
     \begin{align*}
        &\hat x_i(1)-\hat x_{i-1}(1)+\theta\Big(\hat v_i(1)-\hat v_{i-1}(1)\Big)+(\theta+\delta+\\
        &h)u_i(0)-(\theta+\frac{\delta}{2})u_{i-1}(0)-\frac{1}{2}\delta u_{i-1}(t)\\
        &+2\epsilon\theta+\frac{1}{2}\delta\epsilon+\frac{1}{2}\bar a\delta^2+h\bar a\delta-h\epsilon \le 0.
    \end{align*}
    (b) For $t=2,3,\ldots$,
    \begin{align}\nonumber
        &(\frac{3}{2}\delta+h)u_i(t)-(\frac{\delta}{2}+h)u_i(t-1)-\delta u_{i-1}(t)-\\
        &4\epsilon\theta-\frac{3}{2}\epsilon\delta-\frac{1}{2}\bar a\delta^2\le 0. \tag{\text{condition 4}}
    \end{align}
\end{enumerate}
\end{thm}

\begin{proof}
\emph{Condition 3.}
By part (a) in condition 3 we know that there exists a feasible input $u_i(1)$ such that at time step 1 vehicle $i$ can keep a safe distance from its leading vehicle.
Next, we prove that when at time step $t$ vehicle $i$ satisfies part (a), then at time step $t+1$ vehicle $i$ also  satisfies the safety constraint when part (b) in condition 3 holds.
\begin{align*}
    &(\frac{3}{2}\delta+h)u_i(t)-(\frac{\delta}{2}+h)u_i(i-1)-\delta u_{i-1}(t)+\frac{\delta}{2}(\bar a\delta+\epsilon)\\
    &\le 0.
\end{align*}
Since $u_{i-1}(t+1)\ge u_{i-1}(t)-\bar a\delta +\epsilon$, then we have
\begin{align*}
    &\Big(u_i(t)+\epsilon)\delta-\Big(u_i(t-1)+\epsilon)\delta+\frac{\delta}{2}\Big(u_i(t-1)+u_i(t)\Big) 
    \\
    &\le h\Big(u_i(t-1)+\epsilon-\bar a\delta)+\frac{\delta}{2}(u_{i-1}(t)+u_{i-1}(t+1)-2\epsilon)\\
    &-h\Big(u_i(t)+\epsilon-\bar a\delta\Big).
\end{align*}
By condition 1, we have
\begin{align*}
    &\hat x_i(t)+\theta\Big(\hat v_i(t)+u_i(t-1)+\epsilon\Big)+\delta\Big(u_i(t-1)+\epsilon\Big)\\
    &\quad-\frac{1}{2}(\bar a-\frac{\epsilon}{\delta})\delta^2
    \\
    &\le 
    \bar x_{i-1}(t)+\theta\Big(\hat v_{i-1}(t)+u_{i-1}(t-1)-\epsilon\Big)+\frac{\delta}{2}\\
    &\quad\times\Big(u_{i-1}(t-1)+u_{i-1}(t)-2\epsilon\Big)-h\Big(u_i(t-1)+\epsilon-\bar a\delta\Big).
\end{align*}
Then we know that
\begin{align*}
    &\hat x_i(t)+\theta\Big(\hat v_i(t)+u_i(t-1)+\epsilon\Big)+\frac{\delta}{2}\Big(u_i(t-1)+u_i(t)\\
    &\quad+2\epsilon\Big)+\Big(u_i(t)+\epsilon\Big)\delta-\frac{1}{2}(\bar a-\frac{\epsilon}{\delta})\delta^2
    \\
    &\le
    \hat x_{i-1}(t)+\theta\Big(\hat v_{i-1}(t)+u_{i-1}(t-1)-\epsilon\Big)+\frac{\delta}{2}\Big(u_{i-1}(t-1)\\
    &\quad+u_{i-1}(t)-2\epsilon\Big)+\frac{\delta}{2}\Big(u_{i-1}(t)+u_{i-1}(t+1)-2\epsilon\Big)\\
    &\quad-h\Big(u_i(t)+\epsilon-\bar a\delta\Big).
\end{align*}
Considering the limitation of acceleration and deceleration, we have
\begin{align*}
    &\hat x_i(t+1)+\theta\Big(\hat v_i(t+1)+u_i(t)+\epsilon\Big) \\
    &\le
    \hat x_i(t)+\theta\Big(\hat v_i(t)+u_i(t-1)+\epsilon\Big)+\frac{\delta}{2}\Big(u_i(t-1)+u_i(t)\\
    &\quad+2\epsilon\Big)
    \\
    &\hat x_{i-1}(t+1)+\theta\Big(\hat v_{i-1}(t+1)+u_{i-1}(t)+\epsilon\Big) \\
    &\ge
    \hat x_{i-1}(t)+\theta\Big(\hat v_{i-1}(t)+u_{i-1}(t-1)-\epsilon\Big)+\frac{\delta}{2}\Big(u_{i-1}(t-1)\\
    &\quad+u_{i-1}(t)-2\epsilon\Big).
\end{align*}
Therefore, we know that
\begin{align*}
    &\hat x_i(t+1)+\theta\Big(\hat v_i(t+1)+u_i(t)+\epsilon\Big)+\Big(u_i(t)+\epsilon\Big)\delta\\
    &\quad-\frac{1}{2}(\bar a-\frac{\epsilon}{\delta})\delta^2 \\
    &\le
    \hat x_{i-1}(t+1)+\theta\Big(\hat v_{i-1}(t+1)+u_{i-1}(t)-\epsilon\Big)+\frac{\delta}{2}\Big(u_{i-1}(t)\\
    &\quad+u_{i-1}(t+1)-2\epsilon\Big)-h\Big(u_i(t)+\epsilon-\bar a\delta\Big).
\end{align*}
By condition 1 we know that at time step $t+1$ vehicle $i$ also  satisfies the safety constraint which completes the proof.

\noindent\emph{Condition 4.}

By part (a) in condition 4 we know that there exists a feasible input $u_i(1)$ such that at time step 1 vehicle $i$ can keep the minimal allowable distance from its leading vehicle.
Next, we prove that when at time step $t$ vehicle $i$ satisfies part (a), then at time step $t+1$ vehicle $i$ also  satisfies part (a) when part (b) in condition 4 holds.
\begin{align*}
    &(\frac{3}{2}\delta+h)u_i(t)-(\frac{\delta}{2}+h)u_i(t-1)-\delta u_{i-1}(t)-4\epsilon\theta\\
    &-\frac{3}{2}\epsilon\delta-\frac{1}{2}\bar a)\delta^2\le 0.
\end{align*}
Since $u_{i-1}(t+1)\le u_{i-1}(t)+\bar a\delta -\epsilon$, then we have
\begin{align*}
    &-2\theta\epsilon+\Big(u_i(t)+\epsilon)\delta-\Big(u_i(t-1)+\epsilon)\delta+\frac{\delta}{2}\Big(u_i(t-1)\\
    &\quad+u_i(t)-2\epsilon\Big) 
    \\
    &\ge 2\epsilon\theta +2\epsilon\delta+h\Big(u_i(t-1)-\epsilon+\bar a\delta)+\frac{\delta}{2}(u_{i-1}(t)\\
    &\quad+u_{i-1}(t+1)-2\epsilon)-h\Big(u_i(t)-\epsilon+\bar a\delta\Big).
\end{align*}
By condition 2 we have
\begin{align*}
    &\hat x_i(t)+\theta\Big(\hat v_i(t)+u_i(t-1)+\epsilon\Big)+\delta\Big(u_i(t-1)+\epsilon\Big)\\
    &\quad+\frac{1}{2}(\bar a-\frac{\epsilon}{\delta})\delta^2
    \\
    &\ge 
    \bar x_{i-1}(t)+\theta\Big(\hat v_{i-1}(t)+u_{i-1}(t-1)-\epsilon\Big)+\frac{\delta}{2}\Big(u_{i-1}(t-1)\\
    &\quad+u_{i-1}(t)-2\epsilon\Big)-h\Big(u_i(t-1)-\epsilon+\bar a\delta\Big).
\end{align*}
Then we know that
\begin{align*}
    &\hat x_i(t)+\theta\Big(\hat v_i(t)+u_i(t-1)-\epsilon\Big)+\frac{\delta}{2}\Big(u_i(t-1)\\
    &\quad+u_i(t)-2\epsilon\Big)+\Big(u_i(t)+\epsilon\Big)\delta-\frac{1}{2}(\bar a+\frac{\epsilon}{\delta})\delta^2
    \\
    &\ge
    \hat x_{i-1}(t)+\theta\Big(\hat v_{i-1}(t)+u_{i-1}(t-1)+\epsilon\Big)+\frac{\delta}{2}\Big(u_{i-1}(t-1)\\
    &\quad+u_{i-1}(t)+2\epsilon\Big)+\frac{\delta}{2}\Big(u_{i-1}(t)+u_{i-1}(t+1)-2\epsilon\Big)\\
    &\quad-h\Big(u_i(t)-\epsilon+\bar a\delta\Big).
\end{align*}
Considering the limitation of acceleration and deceleration, we have
\begin{align*}
    &\hat x_i(t+1)+\theta\Big(\hat v_i(t+1)+u_i(t)+\epsilon\Big) \\
    &\ge
    \hat x_i(t)+\theta\Big(\hat v_i(t)+u_i(t-1)-\epsilon\Big)+\frac{\delta}{2}\Big(u_i(t-1)+u_i(t)\\
    &\quad-2\epsilon\Big),\\
    &\hat x_{i-1}(t+1)+\theta\Big(\hat v_{i-1}(t+1)+u_{i-1}(t)+\epsilon\Big) \\
    &\le
    \hat x_{i-1}(t)+\theta\Big(\hat v_{i-1}(t)+u_{i-1}(t-1)+\epsilon\Big)+\frac{\delta}{2}\Big(u_{i-1}(t-1)\\
    &\quad+u_{i-1}(t)+2\epsilon\Big).
\end{align*}
Therefore, we know that
\begin{align*}
    &\hat x_i(t+1)+\theta\Big(\hat v_i(t+1)+u_i(t)+\epsilon\Big)+\Big(u_i(t)+\epsilon\Big)\delta\\
    &\quad+\frac{1}{2}(\bar a-\frac{\epsilon}{\delta})\delta^2 \\
    &\ge
    \hat x_{i-1}(t+1)+\theta\Big(\hat v_{i-1}(t+1)+u_{i-1}(t)-\epsilon\Big)+\frac{\delta}{2}\\
    &\quad\times\Big(u_{i-1}(t)+u_{i-1}(t+1)-2\epsilon\Big)-h\Big(u_i(t)-\epsilon+\bar a\delta\Big).
\end{align*}
By condition 2 and condition 3 we know that at time step $t+1$ there also exists a control input $u_i(t+1)$ such that vehicle $i$ also can  keep the minimal allowable distance from its leading vehicle.
\end{proof}

If $\lambda_i(t)=0$, then we have:
\begin{align*}
    u_i(t)=&\theta\Big(\hat v_{i-1}(t)-\hat v_i(t)+u_{i-1}(t-1)-u_i(t-1)\Big) \\ 
     &-2\delta\epsilon- 2\theta\epsilon+\hat x_{i-1}(t)-\hat x_i(t)+\frac{1}{2}\delta\Big(u_{i-1}(t-1) \\
     &-u_i(t-1)+u_{i-1}(t)\Big)/(\frac{1}{2}\delta+h).
\end{align*}

Then we have the controller $\mu$:
\begin{align*}
    u_i(t)&\in U_{feasible}(t)\\
    &=\mu_i(x_i(t-1),v_i(t-1),x_{i-1}(t-1),v_{i-1}(t-1)) \\
    &=\begin{cases}
     \theta\Big(\hat v_{i-1}(t)-\hat v_i(t)+u_{i-1}(t-1)-u_i(t-1)\Big)\ \\ 
     -2\delta\epsilon- 2\theta\epsilon+\hat x_{i-1}(t)-\hat x_i(t)+\frac{1}{2}\delta\Big(u_{i-1}(t-1)\ \\
     -u_i(t-1)+u_{i-1}(t)\Big)/(\frac{1}{2}\delta+h),\ \text{condition 1} \  \\
     \text{and condition 2},\
    \\
    \arg \min{\lambda_i(t)},\ \text{condition 1 and not condition 2},
    \\
    not\ safe, \ \text{else}.
    \end{cases}
\end{align*}

\subsection{Throughput evaluation}
A general evaluation of the intersection's throughput depends on the arrival process (e.g. Poisson process) of vehicles, which is beyond the scope of this paper.
Hence, we consider the worst-case scenario, where vehicles from different directions go through the intersection in alternate. This is the worst case, since the minimal headway between two vehicles in different directions is larger than in the same direction.
We also consider a nominal case where vehicles periodically enter the neighborhood with the nominal crossing speed.

Let $D$ be the crossing time interval for a vehicle, and let $h$ be the minimal headway between two vehicles in different directions. The capacity $F$ of the intersection is lower bounded by
\begin{align}\label{eq:capacity0}
    F\ge\frac{1}{D+h}.
\end{align}

Since the capacity varies with the actual situation, here we only consider a specific situation when all the vehicles are able to satisfy the condition 1 and 2 in Theorem~\ref{thm:1}.
\begin{prp}\label{prp:capacity}
Under the specific situation when all the vehicles are able to satisfy the condition 1 and 2 in Theorem~\ref{thm:1}, capacity F  is lower bounded by:
\begin{align}\label{eq:capacity}
    F \ge \frac{\bar v}{2\epsilon(\theta+\delta)+h\bar v}.
\end{align}
\end{prp}

\begin{figure}[hbt]
    \centering
    \includegraphics[width=0.45\textwidth]{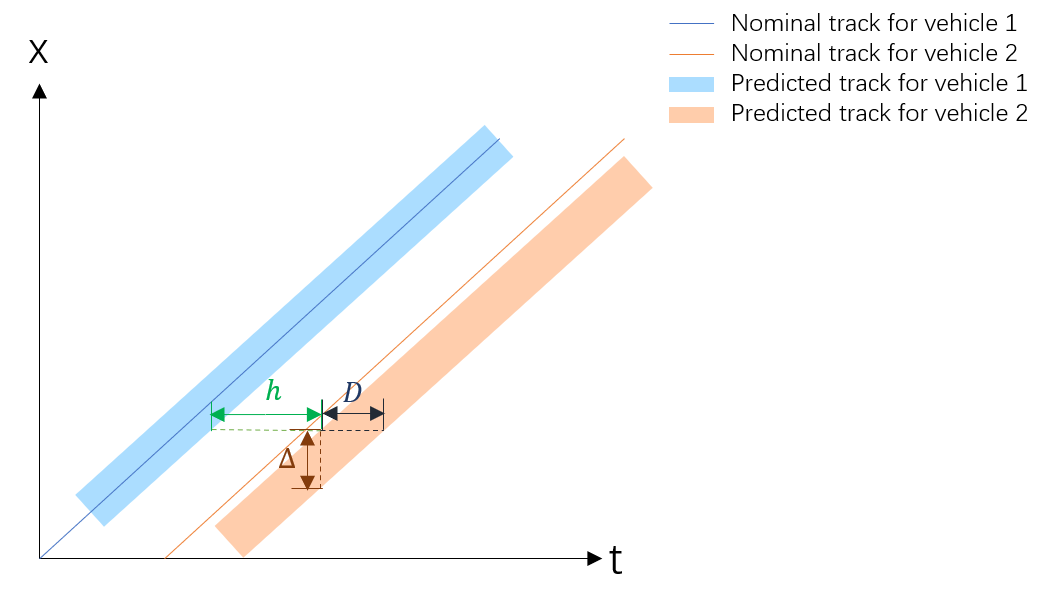}
    \caption{Trajectory sets of two consecutive vehicles.}
    \label{fig:capacity}
\end{figure}

\begin{proof}
By Equation~(\ref{eq:interval}) we have that the interval length of the predicted position at next time step is $2\epsilon(\theta+\delta)$. 
Then we have the $x-t$ curve shown in Figure~\ref{fig:capacity}, in which
$h$ is the distance between the lower bound of the leading vehicle and the upper bound of the following vehicle parallel to the time axis,
$D$ is the mapping of the position interval on the time axis.
Then when $\bar v$ denotes the maximal speed, which is the nominal speed of the leading vehicle, then we have
\begin{align*}
    D=\frac{\Delta}{\bar v}=\frac{2\epsilon(\theta+\delta)}{\bar v}.
\end{align*}
Then by \eqref{eq:capacity0} we have that under the specific situation, capacity $F$  is lower bounded by \eqref{eq:capacity}.
\end{proof}

Note that \eqref{eq:capacity} givens the worst-case throughput, since it considers the case where vehicles from orthogonal directions cross the intersection in alternation.
In practice, vehicles may cross immediately after a leading vehicle in the same direction, which typically leads to a higher throughput.
One way to estimate the typical throughput is to assume Poisson arrivals of vehicles and then study the stochastic stability of the system; see \cite{dai2020}.
 \section{Concluding remarks}\label{sec_conclusion}

In this paper, we solved the robust control problem of vehicle coordination at signal-free intersections with communication latency.
To this end, we designed a controller that coordinates the speed of vehicles which is robust against communication latency and ensures safety.
Based on the controller we discussed the safety criterion both for one time step and time series.
Finally, we studied the relation between latency and intersection capacity under a specific situation.

\bibliographystyle{IEEEtran}
\bibliography{Bibliography}

\end{document}